\documentclass[runningheads]{support/llncs}
\usepackage{graphicx, amsmath, amssymb, enumerate, tikz, mathtools, tikz-cd,pdfcomment, epstopdf,cite,verbatim} 
\usepackage{microtype}
\usepackage[subtle]{savetrees}
\usepackage{wrapfig}
\usepackage[inline]{enumitem}
\usepackage[T1]{fontenc}
\usepackage{algpseudocode,mdframed}
\usepackage{pgfplots, subcaption}
\pgfplotsset{width=5cm,compat=1.9}

\graphicspath{{figures/}}
\renewcommand{\G}{\mathbf{G}}
\newcommand{\F}{\mathbf{F}}
\newcommand{\X}{\mathbf{X}}
\renewcommand{\U}{\mathbf{U}}
\newcommand{\true}{\mathbf{True}}
\newcommand{\false}{\mathbf{False}}

\newcommand{\fcp}{\mathcal{P}}

\newcommand{\fct}{\mathcal{T}}
\newcommand{\sym}{\bigtriangleup}

\renewcommand{\paragraph}[1]{\noindent\textbf{#1}}
\newcommand{\reals}{\mathbb{R}}

\newcommand{\pushcode}[1][1]{\hskip\dimexpr#1\algorithmicindent\relax}

\renewcommand{\tt}{\texttt}
\allowdisplaybreaks

\newcounter{assumption}[section]

\usepackage[colorinlistoftodos]{todonotes}

\title{A Metric for Linear Temporal Logic}

\begin{document}

\author{{\'I}{\~n}igo {\'I}ncer Romeo \and Marten Lohstroh \and
Antonio Iannopollo \and Edward A. Lee \and Alberto Sangiovanni-Vincentelli}
\institute{Department of Electrical Engineering and Computer Sciences \\ University of California, Berkeley CA 94720, USA \\ \email{\{inigo, marten, antonio, eal, alberto\}@berkeley.edu}}

\authorrunning{I. {\'I}ncer Romeo {\it et al.}}

\maketitle

\begin{abstract}
    We propose a measure and a metric on the sets of infinite traces generated by a set of atomic propositions. To compute these quantities, we first map properties to subsets of the real numbers and then take the Lebesgue measure of the resulting sets. We analyze how this measure is computed for 
    Linear Temporal Logic (LTL) formulas. An implementation for computing the measure of bounded LTL properties is provided and explained. This implementation leverages SAT model counting and effects independence checks on subexpressions to compute the measure and metric compositionally.
\end{abstract}

\section{Introduction}




Linear Temporal Logic (LTL), introduced by Amir Pnueli in \cite{pnueliLtl}, allows us to make logical statements over time. This reasoning device is often used to verify whether a system possesses a given temporal property, as well as to synthesize a system with such a property. Applications of temporal logic are vast, ranging from logic design to robotics. In this paper, we consider the following problem: given a temporal logic formula, how large is the set of traces that satisfies it? Given two temporal formulas, are they alike?

Others who have studied this problem \cite{madsen18metric} have shown that quantifying the difference between two specifications can be used, for instance, to evaluate the performance of temporal logic inference algorithms by comparing made inferences to a ground truth. Our own motivation to investigate this problem is its relevance to the design of Cyber-Physical Systems (CPS). ``Late validation'' (i.e., discovering errors or design flaws late in a product's development process or past its deployment) is particularly problematic for CPS because it comes at large human and financial cost; design faults may cause accidents, and solving them may require product recalls or cause manufacturing delays. In an effort to add more rigor to the design process for complex CPS, contract-based design \cite{sangiovanni2012taming, EDA-053} provides tools for reasoning formally about the conditions for correctness of element integration and specification abstraction and refinement.

Using contracts, all components in a design clearly state their environment assumptions and the behaviors they guarantee (with behaviors expressed in formal languages),
allowing for compositional reasoning, also in an automated fashion. 
In \cite{iannopollo16synthesis,iannopollo18decomposition}, Iannopollo \emph{et al.}
describe techniques based on Counterexample-Guided Inductive Synthesis (CEGIS) \cite{Solar-Lezama:2006:CSF:1168857.1168907} to synthesize designs from libraries of contracts, 
satisfying a certain LTL specification. 
Although contracts and their properties help reduce the complexity
of the overall synthesis problem, 
scalability remains elusive, 
as the CEGIS loop depends on LTL model checking, a {\sf PSPACE}-complete problem \cite{sistla85ltl}. 
To tackle complexity, synthesis must be incremental, and we believe an incremental procedure may be within reach soon. For example, {\'I}ncer Romeo {\it et al.} recently introduced means to compute the operation of quotient for assume-guarantee contracts in \cite{incerromeo18quotient}.
Given a specification $\mathcal{C}$ to be implemented, and the specification $\mathcal{C}'$ of a component to be used in the design, the quotient describes the properties that need to be satisfied, in addition to those required by $\mathcal{C}'$, in order to meet $\mathcal{C}$. Therefore, when synthesizing a specification given as an assume-guarantee contract, we can synthesize partial solutions of the design and keep the one whose quotient is the ``smallest'' since we believe the ``size'' of a quotient could be a good indication of complexity. But in order to measure the size of the specification, we need a measure of LTL properties. Introducing that measure is the purpose of this work.


In this paper, we propose a measure for sets of infinite traces spanned by a set of atomic propositions. This measure has the semantics of the amount of trace space that a given property represents. While the measure does not require that properties be specified in LTL, we provide an implementation to compute the measure of bounded LTL formulas and the distance between two such formulas. Our implementation checks whether the measure can be computed compositionally; if not, it computes it using SAT model counting~\cite{biere2009handbook:chap20}. 

In Section \ref{sc:measureIntuitive} we introduce the measure based on counting arguments, and we provide the algorithms for its computation for bounded LTL. We discuss our implementation and experiments in Section~\ref{sc:implementation}. Section~\ref{sc:measureAnalytical} extends our measure to handle infinite traces. In Section \ref{sc:measureAnalytical}, we also show how the counting arguments discussed in Section \ref{sc:measureIntuitive} are a special case of the definition of the measure given in Section \ref{sc:measureAnalytical}. We revisit existing literature in Section~\ref{sc:relWork} and conclude the paper in Section~\ref{sc:conclusion}.

\section{Measuring Properties by Counting Traces}
\label{sc:measureIntuitive}

In this section, we build a measure for bounded LTL formulas based on counting considerations. We argue that a measure that could be used to evaluate the ``size'' of a formula should tell the fraction of trace space represented by the property being measured. Algorithms are presented that compute this measure  compositionally. The counting ideas of this section are generalized to infinite traces in Section \ref{sc:measureAnalytical}.

Since an LTL formula denotes a property (i.e.,\ a set of traces), our purpose is to introduce a measure for sets of traces. Let $\nu$ be our measure and $\phi$ an LTL formula. At a high level, we wish the measure to have the following characteristics:
\begin{itemize}[wide, labelwidth=!]
    \item $\nu(\true) = 1$. The maximum value of the measure is 1 and is achieved when \textit{almost} all possible traces satisfy the formula. The meaning of \emph{almost} will be made clear later. It will coincide with the real-analytic notion of \emph{all traces up to a set of measure zero.}
    \item $\nu(\false) = 0$. The minimum value of the measure is 0 and is achieved when almost no traces satisfy the formula.
    \item A \textit{priori}, we state no preference for a trace over another, i.e., all traces that satisfy the formula should contribute equally to the value of the measure.
    \item The measure is computable in reasonable time.
\end{itemize}

These high-level requirements motivate us to define the measure as \emph{the fraction of the trace space that satisfies the formula}. Suppose that $\beta$ is a Boolean expression. When interpreted as an LTL formula, $\beta$ constrains only the first time step of all traces (i.e., $T = 0$). If, say, $\beta$ is a formula over $n$ atomic propositions, and it is true for $m$ possible combinations of the propositions, the measure of the LTL formula $\beta$ would be $\nu(\beta) = \frac{\textproc{Count}(\beta)}{2^n} = \frac{m}{2^n}$; indeed, the remaining $1 - \frac{m}{2^n}$ of traces begin with a combination of the atomic propositions which does not satisfy $\beta$. Note we assume the existence of a function \textproc{Count} which takes a Boolean expression and outputs the number of satisfying assignments for that expression; in other words, \textproc{Count} carries out SAT model-counting.

\begin{table}
    \centering
    \begin{tabular}{l | c c c c c}
         &  $T = 0$  & $T = 1$  & $T = 2$  & $T = 3$ & $\cdots$\\ \hline
    $x_0$  &  $x_0[0]$ & $x_0[1]$ & $x_0[2]$ & $x_0[3]$ & $\cdots$\\
    $x_1$  &  $x_1[0]$ & $x_1[1]$ & $x_1[2]$ & $x_1[3]$ & $\cdots$\\
    $\vdots$ & $\vdots$ &$\vdots$ &$\vdots$ & $\vdots$ &$\cdots$\\
    $x_{n-1}$  &  $x_{n-1}[0]$ & $x_{n-1}[1]$ & $x_{n-1}[2]$ & $x_{n-1}[3]$ & $\cdots$
    \end{tabular}
    \caption{If defined over $n$ Boolean atomic propositions, traces are uniquely determined by an assignment of values to all variables of the form $x_i[t]$ for all $i \in \{0, \ldots, n-1\}$ and $t \in \omega$.}
    \label{tb:variablesOverTime}
\end{table}

Assume that our LTL formulas are defined over a finite set $\Sigma$ of atomic propositions and that we consider finite traces up to $T = N$ ($N \in \mathbb{N}$). A trace is defined by an assignment of values to all atomic propositions for each time step. As a result, we can interpret an LTL formula as a Boolean expression involving $n$ variables for each time step, as shown in Table~\ref{tb:variablesOverTime}. As an example, suppose our formulas are defined over the propositions $a$ and $b$. Then the formula $\phi = a \land (b \lor X a)$ can be expressed as the Boolean formula $a[0] \land (b[0] \lor a[1])$. We assume we have access to a routine \textproc{AssignedVars} which accepts an LTL formula and returns a set of all atomic propositions that appear in the formula, together with the time step for which they appear in the formula. Applying this procedure to our example yields
$\textproc{AssignedVars}(\phi) = \{(a,0), (a,1), (b,0)\}$. Note that, using this notation, we interpret $(a,i)$ as a different variable from $(b,j)$ whenever $a \ne b$ or $i \ne j$. Let $\phi$ and $\phi'$ be LTL formulas, we also define the function $\textproc{VarUnion}(\phi, \phi') = \textproc{AssignedVars}(\phi) \cup \textproc{AssignedVars}(\phi')$, which returns all time-indexed variables referenced by both formulas, and $\textproc{VarInt}(\phi, \phi')$ for the intersection.
For any pair of LTL formulas, $\phi$ and $\phi'$, $\textproc{AssignedVars}(\phi \land \phi') = \textproc{AssignedVars}(\phi \lor \phi') = \textproc{VarUnion}(\phi, \phi')$.

We define the function $\textproc{VarTimes}$, which accepts an LTL formula and an atomic proposition, and returns the set of time indices for which that atomic proposition appears in the formula. Applying this function to our previous example gives $\textproc{VarTimes}(\phi, a) = \{0,1\}$ and $\textproc{VarTimes}(\phi, b) = \{0\}$.

Finally, we assume there is a bounded version  $\textproc{AssignedVars}(\phi, N)$ that contain the same information as the unbounded version of the function, but with all time indices larger than $N$ removed.

\subsection{Boolean operators}
We explore how the measure behaves with the Boolean operators. Let $\phi$ and $\phi'$ be LTL formulas.
\begin{enumerate}[wide, labelwidth=!]
    \item NOT. The traces that meet $\neg \phi$ are those that do not meet $\phi$. Thus,
\begin{equation}
\label{eq:notMeasure}
\nu(\neg \phi) = 1 - \nu(\phi).
\end{equation}
    \item AND. If a trace $\sigma$ satisfies $\sigma \models \phi \land \phi'$, then
    \begin{align} \label{eq:andreq1}&\sigma \models \phi \quad\text{and} \\ \label{eq:andreq2} &\sigma \models \phi'.\end{align}
    If $\phi$ and $\phi'$ do not share variables, i.e., they do not make assertions over the same atomic proposition at the same time steps, we observe that the requirement that a trace satisfies \eqref{eq:andreq1} is independent of the requirement that it satisfies \eqref{eq:andreq2}. Therefore, the fraction of traces satisfying $\phi \land \phi'$ is the product of the fraction of traces satisfying $\phi$ and the fraction of traces satisfying $\phi'$. On the contrary, if $\phi$ and $\phi'$ share variables, we invoke \textproc{Count}. Thus, the measure of conjunction is given by
    \begin{equation}
    \label{eq:andMeasure}
    \nu(\phi \land \phi') =
    \begin{cases}\nu(\phi) \nu(\phi'), &\text{if } \textproc{VarInt}(\phi, \phi') = \emptyset \\
    \frac{\textproc{Count}(\phi \land \phi')}{2^{|\textproc{VarUnion}(\phi, \phi')|}}, &\text{otherwise}\end{cases}.
    \end{equation}
    Calling \textproc{Count} may be prohibitively expensive if $\phi$ or $\phi'$ comprise a large number of variables. It may thus be useful to compute bounds for the measure of the composition. We can state the following bounds for the measure of an AND of any two LTL formulas:
    \begin{equation}
    \label{eq:andMeasureBounds}
    \nu(\phi) \nu(\phi') \le \nu(\phi \land \phi') \le \min \{\nu(\phi), \nu(\phi')\}.
    \end{equation}
    \item OR. By de Morgan's laws, we can derive the measure of the disjunction from the previous two operations.
    We have $\nu(\phi \lor \phi') = 1 - \nu(\neg \phi \land \neg \phi')$. Thus, the measure of the disjunction of two formulas is
    \begin{equation}
    \label{eq:orMeasure}
    \nu(\phi \lor \phi') =
    \begin{cases}1 - (1 - \nu(\phi))(1 - \nu(\phi')), &\text{if } \textproc{VarInt}(\phi, \phi') = \emptyset \\
    \frac{\textproc{Count}(\phi \lor \phi')}{2^{|\textproc{VarUnion}(\phi, \phi')|}}, &\text{otherwise}\end{cases}.
    \end{equation}
    The following bounds apply to an OR composition for all LTL formulas:
    \begin{equation}
        \begin{split}
    \label{eq:orMeasureBounds}
    \max \{\nu(\phi), \nu(\phi')\} \le \nu(\phi \lor \phi') \quad \text{and} \\
    \nu(\phi \lor \phi') \le 1 - (1 - \nu(\phi))(1 - \nu(\phi')).
        \end{split}
    \end{equation}
\end{enumerate}

Let \textproc{NotMeasure}, \textproc{AndMeasure} and \textproc{OrMeasure} implement, respectively, \eqref{eq:notMeasure}, \eqref{eq:andMeasure}, and \eqref{eq:orMeasure}. We observe that only when $\phi$ and $\phi'$ do not involve the same proposition at the same time steps, the measure of the compositions $\phi \land \phi'$ and $\phi \lor \phi'$ can be computed exactly, just using the measures of $\phi$ and of $\phi'$.
Thus, when computing the measure of an LTL formula, it is advantageous to rewrite the formula to achieve maximum independence.

\paragraph{Example.} Let $\Sigma = \{a, b, c\}$. Let $\phi = a \land \X b$,
$\phi' = b \lor \X a$, and $\phi'' = a \land \neg c$. Consider the formula $\psi = \phi \land \phi' \land \phi''$. Due to considerations of independence,
$\nu(\psi)  =  \nu(\phi')  \nu(\phi \land \phi'')$ since $\phi'$ shares no variables with $\phi$ or with $\phi''$. In this case, the routine \textproc{AndMeasure} would call \textproc{Count} on $\phi \land \phi''$, which makes statements over 3 variables.
Moreover,
$\nu(\psi) \ne \nu(\phi)  \nu(\phi' \land \phi'')$ since $\phi$ shares variables with $\phi' \land \phi''$. Thus, in this case, \textproc{AndMeasure} would call \textproc{Count} on the entire $\psi$, which contains 5 variables.

\subsection{Temporal operators}
\label{sc:measureIntuitiveTemporal}

In principle, the grammar of LTL is the Boolean grammar plus the symbols $\X$ and $\U$. Other temporal operators can be expressed in terms of this grammar. For instance, $\F \phi = \true \,\U\, \phi$ and $\G = \neg \F \neg \phi$. If a bounded semantics is used, we can express $\U$ in terms of $\X$: $\phi \,\U\, \psi = \bigvee_{j = 0}^N \left(\bigwedge_{i = 0}^{j - 1} \X^i \phi\right) \land \X^j \psi$ for $N > 0$ and $\phi \,\U\, \psi = \psi$ for $N = 0$. Note that the effect of $\X$ on a formula consists in adding 1 to all temporal indices of the atomic propositions that appear in the formulas; this is the case because $\X$ commutes with all Boolean (and temporal) operators. As a result, for any LTL formula, we have $\textproc{AssignedVars}(\X \phi) = \{(x, T+1) \,|\, (x,T) \in \textproc{AssignedVars}(\phi) \}$.

We can apply the expressions we obtained for the Boolean formulas to define the relation of the measure and the temporal operators. We will assume the temporal operators are bounded to $T = N$.

\begin{enumerate}[wide, labelwidth=!]
    \item $\X$. If no temporal bound is enforced, applying a next step operator to a formula results in the same measure since the effect of this operator is to rename all variables without causing collisions. We thus have $\nu(\X \phi) = \nu(\phi)$. On the other hand, if a bound \emph{is} enforced, we must provide an interpretation for the case of atomic propositions leaving the bound. In coherence with \cite{de2013linear}, for an atomic proposition $a$ and a time bound $N$, we provide the following semantics for the bounded $\X$ operator: 
    \begin{align*}
    \X^i a = \begin{cases} \X_{\text{LTL}}^i a, &i \le N \\ \false, & \text{otherwise} \end{cases},
    \end{align*}
    where $\X^i_{\text{LTL}}$ is the LTL next operator. For instance, consider the formula $\phi = \X a$ and suppose we impose a time bound $T = 0$. Then no trace can satisfy this formula in our bounded semantics, i.e., the formula is transformed into the empty property. On the other hand, the formula $\phi = \X \neg a$ is transformed into the entire trace space (i.e., all traces satisfy this property).
    \item $\U$. The expansion of the until operator produces several terms with shared dependencies. However, it can be shown that $\phi \,\U\, \psi = \theta_N$, where $\theta_N$ is defined recursively as follows:
    \begin{align}
    \theta_0 = \psi \quad\text{and}\quad \theta_i = \psi \, \lor \, \phi \land \X \theta_{i - 1}.
    \end{align}
    This expansion factors out the applications of the $\X$ operator. 
    In order to compute $\nu\left(\phi \, \U\, \psi\right)$ compositionally, we need two conditions: 
    $\phi$ and $\psi$ must be independent, and their atomic propositions need to be qualified at most at one time step (e.g., $\psi' = b \lor a \land \X a$ does not pass this test because $a$ is qualified by two different time steps). 
    Under the condition of independence, $\nu(\phi \,\U\, \psi) = \gamma_N$, where
    \begin{align}
    \label{eq:untilRecursion}
    \gamma_0 &= \nu(\X^{N} \psi) \quad\text{and} \\
    \nonumber
    \gamma_{i} &= 1 - \left(1 - \nu(\X^{N - i} \psi)\right) \left(1 - \gamma_{i - 1} \cdot \nu(\X^{N - i} \phi)\right).
    \end{align}
\end{enumerate}

A routine for computing the measure of formulas with the $\U$ operator is given below. 
Analogous routines for the remaining temporal operators are special cases of this routine.

\begin{mdframed}
\begin{algorithmic}[1]
\Function{UntilMeasure}{$\phi$, $\psi$, $N$} \Comment{$\U$}
\If {$N = 0$}
\State \Return $\nu(\psi)$
\ElsIf {$\textproc{VarInt}(\phi, \psi) = \emptyset$ and $\forall a \in \Sigma.$ \\
\pushcode[0] $\quad |\textproc{VarTimes}(\phi,a, N)| \le 1$ and \\
\pushcode[0] $\quad |\textproc{VarTimes}(\psi,a, N)| \le 1$}
\State $\gamma \gets \nu(\X^N \psi)$
\For{$i \gets 1, N$}
\State $\alpha \gets \left(1 - \nu(\X^{N - i} \psi)\right)$
\State $\gamma \gets 1 - \alpha \cdot \left(1 - \gamma \cdot \nu(\X^{N - i} \phi)\right)$
\EndFor
\State \Return $\gamma$
\Else
\State $\phi' \gets \bigvee_{j = 0}^N \left(\bigwedge_{i = 0}^{j - 1} \X^i \phi\right) \land \X^j \psi$
\State NumAssertions $\gets \textproc{Count}\left(\phi'\right)$
\State NumVars $\gets |\textproc{AssignedVars}(\phi'; N)|$
\State \Return $\text{NumAssertions} / 2^{\text{NumVars}}$
\EndIf
\EndFunction
\end{algorithmic}
\end{mdframed}

We observe that the given recursion \eqref{eq:untilRecursion} has a fixed point, which allows us to provide the measure for the unbounded case when $\phi$ and $\psi$ do not share variables:
\begin{align}
    \nu(\phi \, \U \, \psi) = \frac{\nu(\psi)}{1 - \nu(\phi)\left(1 - \nu(\psi)\right)}.
\end{align}

\section{Practical Evaluation}
\label{sc:implementation}



We have developed a Python implementation\footnote{\url{https://github.com/icyphy/spec-space}} of the algorithms provided in Section~\ref{sc:measureIntuitive}. Our tool parses LTL formulas and measures them; if a pair of properties is given, it computes their respective distance. After turning the input expression into an abstract syntax tree, we traverse the tree (depth-first, post-order) and annotate visited nodes with dependency information. A subsequent pre-order traversal recursively breaks down the computation of the measure into smaller parts, until a bifurcating node (i.e., $\land$, $\lor$, or $\U$) of which its subtrees
are considered interdependent (see Section~\ref{sc:measureIntuitive}) and thus cannot be measured separately. 
In that case, we invoke \textproc{Count} on an expanded expression obtained from the node with interdependent subtrees, which is carried out by \texttt{sharpSAT}, a DPLL-style \#SAT solver developed by Marc Thurley~\cite{thurley06sharpsat}. The expansion takes care of transforming all temporal operators into Boolean expressions with fresh variables for each time index up to the given time bound.

\subsection{Examples}
Specifications are usually composed out of idioms or templates that denote well-known classes of properties such as safety, liveness, correlation, precedence, and response. Each of these classes are thought to be semantically distinct, and thus should occupy a different fraction of the trace space; is this somehow reflected in their measure? Our results, plotted in Fig.~\ref{fg:mExamples}, confirm this intuition. The measures of these idiomatic LTL expressions are spread remarkably evenly.
\begin{figure}[ht]
    \centering
        {\begin{tikzpicture}
\begin{axis}[
    xlabel=Time bound (N),
    ylabel=$\nu(\phi)$, 
    xtick={0,1, 2, 3, 4, 5},
    legend style={at={(1.1,1)},anchor=north west,font=\tiny}]

\addplot[color=red,mark=diamond,mark options={solid}] coordinates {
    (0, 0)
    (1, 0.25)
    (2, 0.25)
    (3, 0.25)
    (4, 0.25)
    (5, 0.25)
};

\addplot[smooth, color=orange,mark=triangle,mark options={solid}] coordinates {
    (0, 0.5)
    (1, 0.25)
    (2, 0.125)
    (3, 0.0625)
    (4, 0.03125)
    (5, 0.015625)
};

\addplot[smooth,color=brown,mark=star,mark options={solid}] coordinates {
    (0, 0.5)
    (1, 0.5)
    (2, 0.5)
    (3, 0.5)
    (4, 0.5)
    (5, 0.5)
};

\addplot[smooth,color=green,mark=square,mark options={solid}] coordinates {
    (0, 0.75)
    (1, 0.8125)
    (2, 0.890625)
    (3, 0.94140625)
    (4, 0.9697265625)
    (5, 0.984619140625)
};

\addplot[smooth,color=cyan,mark=+,mark options={solid}] coordinates {
    (0, 0.75)
    (1, 0.875)
    (2, 0.9375)
    (3, 0.96875)
    (4, 0.984375)
    (5, 0.9921875)
};

\addplot[smooth, color=blue,mark=oplus,mark options={solid}] coordinates {
    (0, 0.75)
    (1, 0.6875)
    (2, 0.671875)
    (3, 0.66796875)
    (4, 0.6669921875)
    (5, 0.666748046875)
};

\addplot[smooth,color=purple,mark=x,mark options={solid}] coordinates {
    (0, 0.8125)
    (1, 0.828125)
    (2, 0.83203125)
    (3, 0.8330078125)
    (4, 0.833251953125)
    (5, 0.83331298828125)
};

\addplot[smooth, color=violet,mark=otimes,mark options={solid}] coordinates {
    (0, 0.8125)
    (1, 0.71875)
    (2, 0.69921875)
    (3, 0.68505859375)
    (4, 0.67645263671875)
    (5, 0.6717147827148438)
};


\legend{$a \land \X b$, $\G a$, $\G\F a / \F\G a$, $\F a \rightarrow \F b$, $a \rightarrow \F b$, $\G(a \rightarrow \F b)$, $a \rightarrow b \U c$, $\G(a \rightarrow b \U c)$}
\end{axis}
\end{tikzpicture}}
        \caption{Measures up to $N=5$ for ($\Diamond$) initialization, ($\triangle$) safety, ($\star$) liveness/stability, ($\square$) correlation, ($+$, $\oplus$) response, and ($\times$, $\otimes$) precedence.}
        \label{fg:mExamples}
\end{figure}
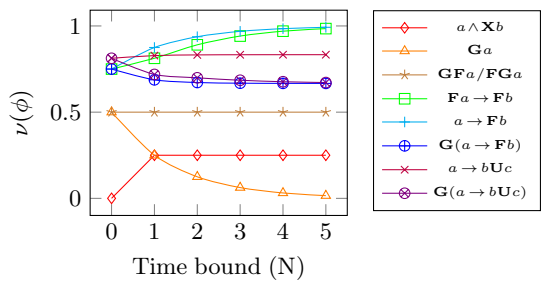
\subsection{Tractability}
Because model counting is a \textsc{\#P}-complete task~~\cite{biere2009handbook:chap20}, unless $P=NP$ there exists no polynomial-time algorithm to implement \textproc{Count}. To make computing our measure more tractable, we employ a divide-and-conquer approach, breaking down the problem into smaller subproblems, and invoke \textproc{Count} only for subexpressions that cannot be broken down further due to interdependencies between variables.
\begin{figure}[ht]
    \centering
    \begin{subfigure}{0.45\columnwidth}
        \resizebox{1.0\columnwidth}{!}
        {\begin{tikzpicture}
\begin{axis}[
    xlabel=Time bound (N),
    ylabel=Execution time ($s$), ymode=log, 
    xtick={0,1, 2, 3, 4, 5},
    legend style={at={(0,1)},anchor=north west,font=\tiny}]

\addplot[smooth, color=green,mark=square,mark options={solid}] coordinates {
    (0,    0.09)
    (1,    0.09)
    (2,   0.10)
    (3,   0.10)
    (4,  0.10)
    (5,  0.11)
};
\addplot[smooth, color=cyan,mark=|,mark options={solid}] coordinates {
    (0,    0.09)
    (1,    0.09)
    (2,   0.10)
    (3,   0.10)
    (4,  0.18)
    (5,  1.02)

};
\addplot[smooth,color=blue,mark=oplus,mark options={solid}] coordinates {
    (0,    0.09)
    (1,    0.09)
    (2,   0.11)
    (3,   0.19)
    (4,  1.73)
    (5,  73.08)
};
\addplot[smooth,color=purple,mark=x,mark options={solid}] coordinates {
    (0,    0.09)
    (1,    0.09)
    (2,   0.12)
    (3,   0.67)
    (4,  36.18)
    (5,  6283.81)
};


\legend{$M=1$, $M=2$, $M=3$, $M=4$}
\end{axis}
\end{tikzpicture}}
        \caption{}
        \label{fg:compExamplea}
    \end{subfigure}
        \begin{subfigure}{0.45\columnwidth}
        \resizebox{0.93\columnwidth}{!}
        {\begin{tikzpicture}
\begin{axis}[
    xlabel=Time bound (N),
    ylabel=Execution time ($s$), 
    ymax=60,
    ytick={0,20, 40},
    xtick={0,10,20,30,40,50},
    legend style={at={(0,1)},anchor=north west,font=\tiny}]

\addplot[smooth, color=green,mark=square,mark options={solid},dotted] coordinates {
(0, 0.33)(10, 1.45)(20, 2.74)(30, 3.78)(40, 4.87)(50, 6.27)
};

\addplot[smooth, color=cyan,mark=|,mark options={solid},dotted] coordinates {
(0, 0.73)(10, 3.20)(20, 6.02)(30, 8.40)(40, 10.90)(50, 13.53)
};

\addplot[smooth, color=blue,mark=oplus,mark options={solid},dotted] coordinates {
(0, 1.41)(10, 6.61)(20, 11.77)(30, 14.84)(40, 19.00)(50, 23.13)
};

\addplot[smooth, color=purple,mark=x,mark options={solid},dotted] coordinates {
(0, 2.06)(10, 8.49)(20, 14.50)(30, 21.09)(40, 27.79)(50, 34.75)
};

\legend{$M=500$, $M=1000$, $M=1500$, $M=2000$}
\end{axis}
\end{tikzpicture}}
        \caption{}
        \label{fg:compExampleb}
    \end{subfigure}
\caption{Computing the measure of $r \rightarrow \F(\bigwedge_{i = 1}^M (x_i \leftrightarrow x_{i+1}))$. (a) Without exploiting independence, the calculation becomes intractable even for small $N$. (b) Computing this measure compositionally scales linearly with $M$ and $N$.}
\label{fig:myfigs}
\end{figure}
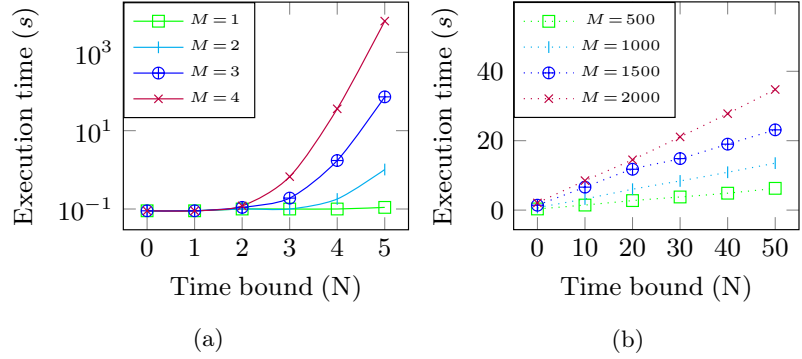
To illustrate the efficacy of this method, let us consider the formula $r \rightarrow \F(\psi)$, where $\psi = \bigwedge_{i = 1}^M (x_i \leftrightarrow x_{i+1})$. The execution time of computing $\nu(r \rightarrow \F(\psi))$ is shown in Figure~\ref{fig:myfigs} for various $M$ and $N$. Because $\psi$ is chain of clauses, each of which shares a variable with the following, all clauses are interdependent. To make matters worse, the expansion of $\F$ yields a disjunction of $N$ time-shifted versions of $\psi$. Applying \textproc{Count} to the resulting Boolean expression is clearly intractable for large $N$, as the number of clauses grows exponentially with $N$. Our algorithm, on the other hand, only has to invoke \textproc{Count} on expressions with $2M$ clauses that jointly comprise a total of $M+1$ distinct variables, and do so $N+1$ times. Because the expansion of $\F$ yields a disjunction of logically equivalent subexpressions, to calculate the measure of $r \rightarrow \F(\psi)$, we invoke \textproc{Count} only once and reuse the result.

\section{Analytical Viewpoint}
\label{sc:measureAnalytical}

In Section \ref{sc:measureIntuitive}, we discussed the measure based on counting considerations. In order to count the satisfying assignments for LTL formulas, we imposed a bound $T = N$ on them. In this section we extend the measure to handle \emph{any} set of infinite traces and is, in fact, not limited to those expressible in LTL. First, we consider a map from the set of all traces in a finite number of atomic propositions to the unit interval $I = [0,1]$. With this map in place, we can interpret any set of traces as a subset of $I$. Thus, we are able to regard properties as subsets of $\reals$. We then introduce the property measure as an integral and provide a calculus for computing it. 

\subsection{Traces as reals}

Suppose $a$ is an atomic proposition. A trace $\sigma$ over $a$ will be given by
$\sigma = a[0], a[1], a[2], \ldots$
Consider the map $\|\cdot\|$ on traces given as follows:
\begin{align}
\label{eq:traceEvaluation}
\|\sigma\| = \sum_{i = 0}^\infty a[i] 2^{-(i+1)}.
\end{align}
Syntactically, this corresponds to the binary number ``0.a[0]a[1]a[2]\ldots,'' i.e., $\|\cdot\|$ prepends the string ``0.'' to the trace $\sigma$, and removes all commas from it. The inverse operation removes the ``0.'' and adds commas between the binary values of the number. As an example, consider the trace $\sigma = 1,0,1,1,1,0,1,0,0,1,\ldots$ The map gives us $\|\sigma\| = 0.1011101001\ldots$.

Let $\Sigma$ be the set of atomic propositions. We assume $\Sigma$ is finite. Then the set of all traces is given by $\fct = \left(2^\Sigma\right)^\omega$, and thus the set of all properties is $\fcp = 2^\fct$. We observe that equation \eqref{eq:traceEvaluation} is a map $\|\cdot\|: \fct \rightarrow I$. This map is clearly surjective, but it is not injective. To see this, consider the following cases: suppose we only have one atomic proposition $a$ and suppose we have the trace $\sigma = a, \neg a, \neg a, \neg a, \ldots$ (i.e., the only trace that satisfies the property $\phi = a \land \X \G \neg a$). Then $\|\sigma\| = 0.1 = 2^{-1}$. Now suppose that $\sigma' = \neg a, a, a, a, \ldots$ (i.e., the only trace that satisfies the property $\phi' = \neg a \land \X \G a$. Then $\|\sigma'\| = \sum_{i = 2}^\infty 2^{-i} = 2^{-1}$. So the real numbers $0.0111...$ and $0.1000...$ are the same. But we think of traces $\sigma$ and $\sigma'$ as being quite different. Indeed, we can build state machines to recognize one and not the other. This shows that $\|\cdot\|$ is not injective.

When $\Sigma$ contains $n$ atomic propositions instead of 1, traces are given in the form shown in Table~\ref{tb:variablesOverTime}. We observe that we can \emph{linearize} the trace to
\[
x_0[0], x_1[0], \ldots, x_{n-1}[0], x_0[1], \ldots, x_{n-1}[1],
x_0[2] \ldots
\]
This interpretation allows us to extend $\|\cdot\|$ to a map from sets of infinite traces over any finite set of atomic propositions to the unit interval.

\subsection{Properties as subsets of $\reals$}

Let $\phi$ be a property over $\Sigma$. We define a map $f: \fcp \rightarrow 2^I$ from properties to subsets of the reals given by
$f \phi = \{ \|\sigma\| : \sigma \in \phi\}$. We observe that $f$ allows us to interpret $\phi$ as subset of $\reals$. We call $f\phi$ the \emph{property intervals} of $\phi$.

\subsection{A measure for properties}

Now that we can interpret properties as subsets of the reals, we proceed to define the property measure and show that it meets the requirements of a measure. We also show that it absorbs as a special case the measure described in Section \ref{sc:measureIntuitive}.
\begin{definition}
The property measure $\nu: \fcp \rightarrow I$ is given by
\begin{align}
\label{eq:propMeasure}
\nu = \mu \circ f,
\end{align}
where $\mu$ is the Lebesgue measure.
\end{definition}

\begin{theorem}
$\nu$ is a well-defined measure.
\end{theorem}
\begin{proof}
Since $\mu$ is nonnegative, $\nu$ is nonnegative. Moreover, we have $\nu(\emptyset) = \mu \circ f (\emptyset) = \mu(\emptyset) = 0$. Finally, we have to verify countable additivity. Let $\phi, \phi' \in \fcp$ such that $\phi \cap \phi' = \emptyset$. We observe that irrational numbers have a unique binary representation. Therefore, if $\|\cdot\|$ maps more than one trace to the same real number, this number must be rational. We compute the measure:
$
    \nu(\phi \cup \phi') =  \mu \circ f(\phi \cup \phi') =
    \mu \left( \left(f(\phi) \cup f(\phi')\right) \cap \left( I_q \cup I_{q}^c \right) \right),
$
where $I_q$ and $I_q^c$ are respectively the rational and irrational numbers in $I$. The measure of $\phi \cup \phi'$ becomes
\begin{align*}
    &\mu \left( f(\phi) \cap I_q^c \cup f(\phi') \cap I_q^c \cup \left( f(\phi) \cup f(\phi')\right) \cap I_q \right) = \\
    &\mu \left( f(\phi) \cap I_q^c \right) +
    \mu \left( f(\phi') \cap I_q^c \right) \\ &+
    \mu \left( \left(f(\phi) \cup f(\phi')\right) \cap I_q \right).
\end{align*}
Since the Lebesgue measure of a countable set is zero, we have
\begin{align*}
\nu(&\phi \cup \phi') = \mu \left( f(\phi) \cap I_q^c \right) + \mu \left( f(\phi') \cap I_q^c \right) \\ &=
\mu \left( f(\phi) \cap I_q^c \right) + \mu \left( f(\phi') \cap I_q^c \right) \\ &+
\mu \left( f(\phi) \cap I_q \right) + \mu \left( f(\phi') \cap I_q \right) \\ &=
\mu \left( f(\phi) \right) + \mu \left( f(\phi') \right) =
\nu(\phi) + \nu(\phi').
\end{align*}
Countable additivity follows by induction.
\end{proof}

Equation \eqref{eq:propMeasure} can be rewritten as $\nu(\phi) = \int_{f\phi} d\mu$. This expression allows us to easily generate new measures that weight traces unequally. Suppose $g$ is a measurable function from $I$ to the nonnegative reals with $\int_{x \in I} g(x) dx = 1$. Then $g$ can be used to weight traces in order to attribute different importance to different traces. The weighted measure is given by
\begin{align}
\label{eq:propWeightedMeasure}
\nu(\phi) = \int_{f\phi} g\, d \mu.
\end{align}
We leave further research of weighted metrics for future work. Now we proceed to construct a metric using the measure, but first we show that this notion of the measure generalizes the discussion of Section \ref{sc:measureIntuitive}.

\paragraph{Bounded case.} Suppose $\phi$ is an LTL formula defined for a bound $T = N$ over a finite set $\Sigma$ of atomic propositions. While introducing the measure $\nu$ through counting considerations in Section \ref{sc:measureIntuitive}, we first interpreted the given formula using the bounded semantics described in Section \ref{sc:measureIntuitiveTemporal}, thus generating a new formula $\phi'$ which makes statements only over the first $N + 1$ time steps of the traces. Suppose $\phi'$ has $m$ satisfying assignments $\sigma_1, \ldots,\, \sigma_m$, where the $\sigma_i$ are traces of length $N+1$. Thus, we can write $\phi' = (\sigma_i)_{i = 1}^m$. Let $\phi''$ be a formula syntactically equal to $\phi'$, but interpreted over infinite traces. Since $\phi'$ only makes statements over the first $N+1$ time steps, any infinite extension of $\sigma_i$ satisfies $\phi''$. Let $\sigma_i' = \{\gamma \,|\, \gamma \in \fct \text{ and } \gamma[0\ldots N] = \sigma_i\}$; that is, $\sigma_i'$ contains all infinite extensions of the bounded trace $\sigma_i$. Then, $\gamma \models \phi''$ for all $\gamma \in \sigma_i'$ for $1 \le i \le m$. Moreover, in this case we have $\phi'' = \cup_{i = 1}^m \sigma_i'$. The $\sigma_i'$ are disjoint, so $\nu(\phi'') = \sum_{i = 1}^m \nu(\sigma_i') = m \,2^{-(N+1)|\Sigma|}$, which matches the counting interpretation of the measure discussed in Section \ref{sc:measureIntuitive}.

\subsection{A metric for properties}

For any two sets $X$ and $Y$, we let their symmetric difference be given by $X \sym Y = (X - Y) \cup (Y - X)$. We now define a function that we use to build a metric.

\begin{definition}
    Let $\phi, \phi' \in \fcp$. The \emph{property distance} $d : \fcp^2 \rightarrow I$ is given by $d(\phi, \phi') = \nu(\phi \sym \phi')$.
\end{definition}

\begin{proposition}
\label{prop:triangle}
    $d$ satisfies the triangle inequality.
\end{proposition}
\begin{proof}
    Let $\phi, \phi' \in \fcp$. We provide some intermediate results:
    \begin{enumerate}[label=\alph*.]
        \item Here we show that $\phi \subseteq \phi' \Rightarrow \nu(\phi) \le \nu(\phi')$;
        Suppose $\phi \subseteq \phi'$; then $\nu(\phi') = \nu\left((\phi' - \phi) \cup \phi\right) = \nu(\phi' - \phi) + \nu(\phi) \ge \nu(\phi)$ because $\nu$ is nonnegative.
        \item Now we wish to show that $\nu(\phi \cup \phi') \le \nu(\phi) + \nu(\phi')$. We observe that
        \begin{align*}
            \nu(\phi \cup \phi') &= \nu\left((\phi - \phi') \cup (\phi' - \phi) \cup \phi \cap \phi' \right)\\ &=
        \nu(\phi - \phi') + \nu(\phi' - \phi) + \nu(\phi \cap \phi') \\ &\le \nu(\phi - \phi') + \nu(\phi' - \phi) + 2 \nu(\phi \cap \phi') \\ &= \nu(\phi) + \nu(\phi').
        \end{align*}
    \end{enumerate}
Let $\psi \in \fcp$. We have $\phi \sym \phi' \subseteq (\phi \sym \psi) \cup (\phi' \sym \psi)$. Therefore, applying (a) followed by (b), we obtain
$d(\phi, \phi') \le d(\phi, \psi) + d(\phi', \psi)$.
\end{proof}

\begin{proposition}
    $d$ does not satisfy the strong triangle inequality.
\end{proposition}
 \begin{proof}
     Let $\phi, \phi', \phi'' \in \fcp$. The strong triangle inequality, or ultrametric inequality, requires $d(\phi, \phi') \le \max\{d(\phi, \phi''), d(\phi', \phi'')\}$.
     Suppose that $\phi''$ is empty and $\phi$ and $\phi'$ are nonempty, disjoint, and have positive measure. Then $d(\phi, \phi') = \nu(\phi) + \nu(\phi') > \max\{\nu(\phi), \nu(\phi')\} = \max\{d(\phi, \phi''), d(\phi', \phi'')\}$. This counterexample proves the proposition.
 \end{proof}

Let $\bar\fcp = \fcp/R$, where $R = \{(\phi, \phi') \in \fcp^2 \,| \, d(\phi, \phi') = 0 \}$; that is, the equivalence classes of $\bar\fcp$ consist of properties whose pairwise distance is zero.

\begin{definition}
    Let $\bar \phi, \bar \phi' \in \bar \fcp$. We define the property metric $\bar d: {\bar \fcp}^2 \rightarrow \mathbb{R} $ as $d(\bar \phi, \bar \phi') = d(\phi, \phi')$, where $\phi \in \bar \phi$ and $\phi' \in \bar \phi'$.
\end{definition}

\begin{proposition}
    $\bar d$ is a metric.
\end{proposition}
\begin{proof}
    First we have to show that $\bar d$ is well-defined. Let $\bar \phi$ and $\bar \phi'$ be equivalence classes in $\bar \fcp$ and let $\phi, \psi \in \bar \phi$ and $\phi' \in \bar \phi'$. We have
    \begin{align*}
        d(\phi, \phi') &\le d(\phi, \psi) + d(\phi', \psi) = d(\phi', \psi)
        \\ &\le
        d(\phi', \phi) + d(\phi, \psi) = d(\phi', \phi)
        \Rightarrow d(\phi, \phi') = d(\phi', \psi).
    \end{align*} 
    Thus, $\bar d$ is independent of the members of the equivalence classes used in its computations. It follows that $\bar d$ is well-defined. Moreover, $\bar d$ is nonnegative because $\nu$ is nonnegative; it satisfies the triangle inequality because $d$ satisfies it (see Proposition \ref{prop:triangle}); it is symmetric due to the symmetry of $\sym$. It remains to be shown that $\bar d(\bar \phi, \bar \phi') = 0 \Leftrightarrow \bar \phi = \bar \phi'$.

    Suppose $d(\bar \phi, \bar \phi')=0$ and $\phi \in \bar \phi$ and $\phi' \in \bar \phi'$. Then $d(\phi, \phi') = 0$. But this means that $\phi$ and $\phi'$ belong to the same equivalence class in $\bar \fcp$, so $\bar \phi = \bar \phi '$. Conversely, suppose that $\bar \phi = \bar \phi '$ and let $\phi \in \bar \phi$ and $\phi' \in \bar \phi' = \phi$. Then $\phi$ and $\phi'$ belong to the same equivalence class in $\bar \fcp$, which means that $d(\phi, \phi') = 0$; therefore, $\bar d(\bar \phi, \bar \phi') = 0$. 
\end{proof}

In consequence, we can use $d$ to measure the distance between two given LTL properties. Moreover, the distance between two formulas will be zero iff they differ by a property of measure zero.

\subsection{Calculating with the property measure}
\label{sc:graphicalCalculus}

Suppose $\phi$ is defined over a single atomic proposition $a$ and $\phi = a$. Then $\phi$ contains all traces that start with a 1. Clearly, $f\phi = [0.5, 1]$. Therefore, application of \eqref{eq:propMeasure} results in $\nu(\phi) = 0.5$.

Suppose $\phi = a \land \X a \lor \neg \X^2 a$. We observe that $\nu(\X a)$ contains all numbers whose second digit is 1. Thus, $\nu(\X a) = [1/4,1/2] \cup [3/4,1]$. Figure \ref{fg:measureExample} shows the intervals to which each of the terms in $\phi$ map. We observe that $a$ has one segment of size $2^{-1}$; $\X a$ has two segments of size $2^{-2}$; and $X^2 a$ has 4 segments of size $2^{-3}$. The property $a \land \X a$ intersects the segment of $a$ with one of the segments of $\X a$, as shown in Figure \ref{fg:measureExampleb}. Thus, we have $\nu(a \land \X a) = 2^{-2}$. Finally, we observe that we have to add three of the 4 segments of $\X^2 a$ to get the final answer, and we obtain $\nu(\phi) = 2^{-2} + 3 \cdot 2^{-3} = 5/8$.

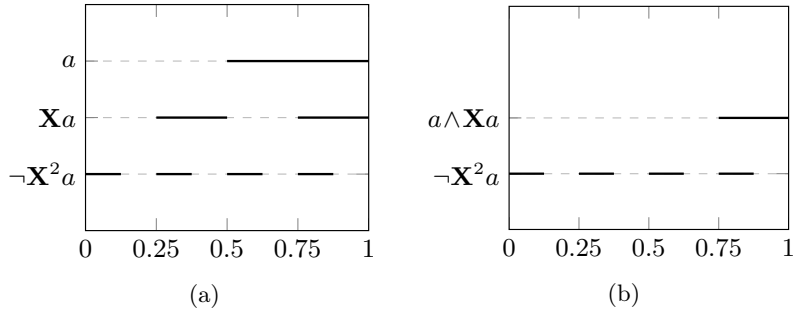
\begin{figure}[h]
    \centering
    \begin{subfigure}{0.45\columnwidth}
        \resizebox{0.95\columnwidth}{!}
        {    \begin{tikzpicture}
\begin{axis}[
    title={},
    xlabel={},
    ylabel={},
    xmin=0, xmax=1.0,
    ymin=0, ymax=2,
    xtick={0,0.25, 0.5, 0.75, 1.0},
    ytick={0.5, 1.0, 1.5},
    yticklabels={$\neg\X^2 a$, $\X a$, $a$},
    legend pos=north west,
    ymajorgrids=true,
    grid style=dashed,
]
 
\addplot[
    color=black,
    thick,
    ]
    coordinates {
    (0.5,1.5)(1.0,1.5)
    };
    
\addplot[
    color=black,
    thick,
    ]
    coordinates {
    (0.75,1.0)(1.0,1.0)
    };
    
\addplot[
    color=black,
    thick,
    ]
    coordinates {
    (0.25,1.0)(0.5,1.0)
    };
    
\addplot[
    color=black,
    thick,
    ]
    coordinates {
    (0.0,0.5)(0.125,0.5)
    };
\addplot[
    color=black,
    thick,
    ]
    coordinates {
    (0.25,0.5)(0.375,0.5)
    };
\addplot[
    color=black,
    thick,
    ]
    coordinates {
    (0.5,0.5)(0.625,0.5)
    };
\addplot[
    color=black,
    thick,
    ]
    coordinates {
    (0.75,0.5)(0.875,0.5)
    };
    \legend{}
 
\end{axis}
\end{tikzpicture}}
        \caption{}
        \label{fg:measureExamplea}
    \end{subfigure}
        \begin{subfigure}{0.45\columnwidth}
        \resizebox{0.95\columnwidth}{!}
        {    \begin{tikzpicture}
\begin{axis}[
    title={},
    xlabel={},
    ylabel={},
    xmin=0, xmax=1.0,
    ymin=0, ymax=2,
    xtick={0,0.25, 0.5, 0.75, 1.0},
    ytick={0.5, 1.0},
    yticklabels={$\neg \X^2 a$, $a \land \X a$},
    legend pos=north west,
    ymajorgrids=true,
    grid style=dashed,
]

\addplot[
    color=black,
    thick,
    ]
    coordinates {
    (0.75,1.0)(1.0,1.0)
    };
    
\addplot[
    color=black,
    thick,
    ]
    coordinates {
    (0.0,0.5)(0.125,0.5)
    };
\addplot[
    color=black,
    thick,
    ]
    coordinates {
    (0.25,0.5)(0.375,0.5)
    };
\addplot[
    color=black,
    thick,
    ]
    coordinates {
    (0.5,0.5)(0.625,0.5)
    };
\addplot[
    color=black,
    thick,
    ]
    coordinates {
    (0.75,0.5)(0.875,0.5)
    };
    \legend{}
 
\end{axis}
\end{tikzpicture}}
        \caption{}
        \label{fg:measureExampleb}
    \end{subfigure}
\caption{Computing the measure of the property $a \land \X a \lor \neg \X^2 a$.}
\label{fg:measureExample}
\end{figure}

\section{Related Work}
\label{sc:relWork}
A \emph{robust} semantics of LTL
has been presented by Tabuada \textit{et al.} in \cite{tabuada15robustltl} which formalizes the idea
that a trace that violates a property only finitely many times is preferable over one that always violates it;
the trace that is in lesser violation must therefore be ``closer'' to satisfying the specification. For LTL 
Assume/Guarantee properties ($\varphi \Rightarrow \psi$), a \textit{da Costa} algebra is used to compute how much a trace is allowed to diverge from $\psi$ given a violation of $\varphi$.
Other papers, \cite{jaksic16stldistance}, \cite{dealfaro04quantitative}, by defining metrics, 
also investigate the relation between traces and temporal properties.
The main difference between these papers and our approach is that we focus solely on specifications, 
irrespective of particular implementations and the degree to which they satisfy those specifications.

In \cite{finkbeiner14counting}, Finkbeiner \textit{et al.} explore the theoretical complexity of the problem of LTL model counting for safety formulas, 
extended to full LTL by Torfah \textit{et al.} \cite{Torfah2018}.
They distinguish between two classes of models, i.e., word- and tree-models of bounded size. 
In this paper, we focus on sets of traces, which closely relate to the problem of word-model counting. 
It is shown in \cite{finkbeiner14counting} that this counting problem is \textsf{\#PSPACE}-complete when a binary-encoded bound is used, 
and they provide a counting algorithm that is linear in the temporal bound but double-exponential in the size of the formula.
Although the complexity of the problem in the worst case lies within this theoretical framework, our compositional approach is more tractable in practice.
Indeed, by invoking a SAT model counter on smaller problem instances, we can compute the measure of longer properties,
which are hard to analyze using the techniques introduced by Finkbeiner \textit{et al.}

The work that we find most closely related to our own is a recent pre-print by Madsen \textit{et al.} \cite{madsen18metric},
which proposes two metrics for Signal Temporal Logic (STL) properties which are assumed rectangular and belonging to a compact vector space.
One is based on the \textit{Hausdorff} distance, and
the other is based on the Symmetric Difference (SD) between two properties.
Albeit developed independently, our LTL metric shares significant similarities with their SD metric.
There are, however, some fundamental differences.
Madsen \textit{et al.} focus on STL, and formulates some restrictive assumptions to ensure the metrics and the proposed algorithms are correct.
We focus on LTL, and do not make any such assumptions.
Moreover, the Lebesgue measure used in the SD metric is applied directly to properties, 
whereas we first transform the property into a subset of $I$ and take the Lebesgue measure afterwards.
Lastly, unlike the SD metric, our metric can be computed compositionally provided the necessary independence conditions.

\section{Conclusions}
\label{sc:conclusion}

This paper provides a measure and a metric for sets of infinite traces satisfying a given LTL property (referred to as measure of the LTL formula) defined over a finite set of atomic propositions. We have shown how the measure behaves for bounded LTL properties, and provided an implementation that computes the measure of bounded LTL properties based on considerations of independence and on model counting. We plan to continue developing measure composition rules to handle common properties. Now that we have a measure and a metric available, we plan to use it to implement greedy approaches for LTL synthesis.

Future theoretical work includes developing the graphical calculus outlined in Section \ref{sc:graphicalCalculus}. This may enable more efficient means for computing the property measure and for carrying out LTL model counting. Moreover, the property intervals we use to compute the measure may have implications for model checking and synthesis. Indeed, we showed that these intervals simultaneously encode the measure and the logical structure of the properties. Moreover, the map $f: \fcp \rightarrow 2^I$ mapping properties to subsets of the unit interval possibly is also capable of mapping morphisms of properties to morphisms of subsets of $I$. This would turn $f$ into a functor. Researching this could unveil additional algebraic structure that may be possible to import into the world of properties via tools from category theory. Finally, we have shown how the property measure can be extended to weighted property measures, but we leave the explorations of specific costs for future work.

\section*{Acknowledgments}

We gratefully acknowledge valuable conversations with Tiziano Villa and Daniel Fremont during the preparation of this manuscript. 

The work in this paper was supported in part by the National Science
Foundation (NSF), awards \#CNS-1739816 (Quantitative Contract-Based Design Synthesis and Verification for CPS Security) and
\#CNS-1836601 (Reconciling Safety with the Internet), and the iCyPhy Research Center (Industrial Cyber-Physical
Systems), supported by Avast, Camozzi Group, Denso, Ford, and Siemens.

\bibliographystyle{support/splncs04}
\bibliography{support/references}

\end{document}